\RequirePackage{fix-cm}
\documentclass[smallextended]{svjour3}       
\smartqed  
\usepackage{amssymb,amsmath,amsfonts,mathrsfs}
\usepackage[utf8]{inputenc}
\usepackage[english]{babel}
\begin{document}

\title{On Integer Programming With Almost Unimodular Matrices And
The Flatness Theorem for Simplices
}


\author{D. V. Gribanov         \and
        S. I. Veselov 
}


\institute{ D. V. Gribanov \at Lobachevsky State University of Nizhny Novgorod, 23 Gagarin Avenue, Nizhny Novgorod, Russian Federation, 603950,\\ Laboratory of Algorithms and Technologies for Networks Analysis, National Research University Higher School of Economics, 136 Rodionova, Nizhny Novgorod, Russian Federation, 603093,\\ 
\email{dimitry.gribanov@gmail.com}
\and
S. I. Veselov \at Lobachevsky State University of Nizhny Novgorod, 23 Gagarin Avenue, Nizhny Novgorod, Russian Federation, 603950,\\ \email{veselov@vmk.unn.ru}
}

\date{Received: date / Accepted: date}

\maketitle

\begin{abstract}
Let $A$ be an $(m \times n)$ integral matrix, and let $P=\{ x : A x \leq b\}$ be an $n$-dimensional polytope. The width of $P$ is defined as $w(P)=\min\{ c \in \mathbb{Z}^n\setminus\{0\} :\: \max\limits_{x \in P} c^\top x - \min\limits_{x \in P} c^\top x \}$. 
Let $\Delta(A)$ and $\delta(A)$ denote the greatest and the smallest absolute values of a determinant among all $r(A) \times r(A)$ sub-matrices of $A$, where $r(A)$ is the rank of the matrix $A$.

We prove that if every $r(A) \times r(A)$ sub-matrix of $A$ has a determinant equal to $\pm \Delta(A)$ or $0$ and $w(P)\ge (\Delta(A)-1)(n+1)$, then $P$ contains  $n$ affine independent integer points. Additionally, we have similar results for the case of \emph{$k$-modular} matrices. The matrix $A$ is called \emph{totally $k$-modular} if every square sub-matrix of $A$ has a determinant in the set $\{0,\, \pm k^r :\: r \in \mathbb{N} \}$.

When $P$ is a simplex and $w(P)\ge \delta(A)-1$, we describe a polynomial-time algorithm for finding an integer point in $P$.

Finally, we show that if $A$ is \emph{almost unimodular}, then integer program $\max \{c^\top x :\: x \in P \cap \mathbb{Z}^n \}$ can be solved by a polynomial-time algorithm. The matrix $A$ is called \emph{almost unimodular} if $\Delta(A) \leq 2$ and any $(r(A)-1)\times(r(A)-1)$ sub-matrix has a determinant from the set $\{0,\pm 1\}$.

\end{abstract}

\section{Introduction}
Let $A$ be an $m\times n$ integral matrix. Its $ij$-th element is
denoted by $A_{i\,j}$, $A_{i\,*}$ is the $i$-th row of $A$, and
$A_{*\,j}$ is the $j$-th column of $A$. For a vector $b\in
\mathbb{Z}^{n}$, by $P(A,b)$ (or by $P$) we denote the polyhedron
$\{ x \in \mathbb{R}^{n} : A x \leq b\}$. The set of all vertices of a polyhedron $P$ is denoted by $vert(P)$. 

Let $r(A)$  be the rank of the matrix $A$.
Let $\Delta(A)$ and $\delta(A)$ denote the greatest and the smallest absolute values of the determinant among all $r(A) \times r(A)$ sub-matrices of $A$. Let $\Delta_{lcm}(A)$ and $\Delta_{gcd}(A)$ be the least common multiple and greatest common divisor (resp.) of absolute values of a determinant among all $r(A) \times r(A)$ sub-matrices of $A$. 

Now we refer to the notion of the \emph{$k$-modular} matrices that has been introduced in \cite{Kotnyek}. The matrix $A$ is called totally $k$-modular if every square sub-matrix of $A$ has a determinant in the set $\{0,\, \pm k^r :\: r \in \mathbb{N} \}$. The matrix $A$ is called \emph{$k$-modular} if every $r(A) \times r(A)$ sub-matrix of $A$ has a determinant in the set $\{0,\, \pm k^r :\: r \in \mathbb{N} \}$.

Also we refer to the notion of \emph{almost unimodular} matrices that was introduced in \cite{PadbergConj_AlmostUnimod} for square case. The matrix $A$ is called \emph{almost unimodular} if $\Delta(A) \leq 2$ and any $(r(A)-1)\times(r(A)-1)$ sub-matrix has a determinant from the set $\{0,\pm 1\}$

For a matrix $B \in \mathbb{R}^{s \times n}$, $cone(B) = \{x:\: x =
B t,\, t \in \mathbb{R}^{n},\, t_i \geq 0 \}$ is a cone spanned by
columns of $B$ and $conv(B) = \{x:\: x = B t,\, t \in
\mathbb{R}^{n},\, t_i \geq 0,\, \sum_{i=1}^{n} t_i = 1  \}$ is the
convex hull spanned by columns of $B$.

For a vertex $v$ of $P(A,b)$, $N(v) = \{ x \in \mathbb{R}^n :\: A_{J\,*} x \leq b \}$, where $J = \{k :\: A_{k\,*} v = b \}$.

The following theorem was proved in  \cite{VeselovChirkovBimodEn2008}
\begin{theorem}\label{MainTh1} If every $ n\times n$ determinant of $A$ belongs to $\{-2,-1,0,1,2\}$ and $P(A,b)$ is full-dimensional, then  

1. $P(A,b) \cap \mathbb{Z}^n \not= \emptyset $.

2. One can check the emptiness of the set $P(A,\,b) \cap \mathbb{Z}^n$ by a polynomial-time algorithm.

3. For every row $a$ of $A$,
the problem  $max\{a^\top x : A x \leq b,\  x \in \mathbb{Z}^n \}$ can be solved by a polynomial-time algorithm.

4. For every $v \in vert(conv(P(A,b) \cap \mathbb{Z}^n))$ there exists $u \in vert(P)$, such that $v$ lies on some edge which contain $u$ ($v$ lies on an edge of $N(u)$).

5. If each $n \times n$ sub-determinant of $A$ is not equal to zero, then the problem $max\{c^\top x : A x \leq b,\  x \in \mathbb{Z}^n \}$ can be solved by a polynomial time algorithm.

\end{theorem}
The remarkable result was obtained by V.E. Alekseev and D.V.
Zakharova in \cite{AlekseevZaharova} for $\{ 0,1 \}$-matrices.

\begin{theorem}\label{MainThMalishev}
Let $A \in \{ 0,1 \}^{m \times n}$, $b\in \{ 0,1\}^m$, $c\in
\{0,1\}^n$. Let all rows of $A$ have at most
2 ones. Then, the problem $max\{c^{\top}x : A x \leq b,\  x \in
\mathbb{Z}^n \}$ can be solved by a polynomial-time algorithm when $\Delta (\binom{~c^\top}{A})$ is fixed.
\end{theorem}

Width of $P$ is defined as
$$ w(P)=\min\{c\in \mathbb{Z}^n\setminus\{0\} :\: \max\limits_{x \in P} c^\top x - \min\limits_{x \in P} c^\top x \}.$$

Now we refer to the classical flatness theorem due to Khinchine \cite{Khinchine}. Let $P$ be a convex body. Khinchine shows that if $P \cap \mathbb{Z}^n = \emptyset$, then $w(\,P\,) \leq f(n)$, where $f(n)$ is a value that depends only on a dimension. There are many estimates on $f(n)$ in the works \cite{FlatnThviaLocalThOfBanach,Banasz,Dadush,KannLovasz,Khinchine,Rudelson}. There is a conjecture claiming that $f(n) = O(n)$ \cite{Dadush,KannLovasz}. The best known upper bound on $f(n)$ is $O(n^{4/3}\, log^c(n))$ due to Rudelson \cite{Rudelson}, where $c$ is some constant that does not depend on $n$.

The paper \cite{Gribanov_General} contains an estimate of the width for a special class of polytopes.

\begin{theorem}\label{StrictKMod_PolyWidthT}
Let $A \in \mathbb{Z}^{m \times n}$, $b \in \mathbb{Z}^{m}$, $P(A,b)$ be a polytope and every $r(A) \times r(A)$ sub-determinant of matrix $A$ is equal to $\pm \Delta(A)$ or 0. If $w(P(A,b)) > (\Delta(A)- 1)\,(n+1)$, then $|P(A,b) \cap \mathbb{Z}^{n}| \geq n + 1$. Moreover we can find an integer point in $P(A,b) \cap \mathbb{Z}^{n}$ using a polynomial time algorithm.
\end{theorem}

We give here another proof of this result.

An interesting problem is estimating of $f(n)$ for empty lattice simplices \cite{FlatnThviaLocalThOfBanach,Haase_Ziegler,Kantor,IntroToEmptyLatticeSimpl}. A simplex $S$ is called empty lattice if $vert(S) \subseteq \mathbb{Z}^n$ and $S \cap \mathbb{Z}^n \setminus vert(S) = \emptyset$. Best known estimate of $f(n)$ for \emph{ the empty lattice} simplices is $O(n\,log(n))$ due to \cite{FlatnThviaLocalThOfBanach}. 

In this paper we will prove that the width of a simplex (not necessary with integer vertices) without lattice points is at most $\delta(A) - 1$, where $A$ is the restriction matrix of the simplex. Moreover, if its width is at least $\delta(A) - 1$, then we can find an integer point in the simplex by a polynomial-time algorithm presented in this paper. 

The authors consider this paper as a part of general problem for finding out critical values of parameters, when a given problem changes complexity. For example, the integer programming problem is polynomial time solvable on polyhedrons with integer vertices, due to \cite{Khachiyan}. On the other hand, it is NP-complete in the class of polyhedrons with denominators of extreme points equal $1$ or $2$, due \cite{Padberg}. The famous $k$-satisfiability problem is polynomial for $k \leq 2$ but is NP-complete for all $k > 2$. In the papers \cite{Malyshev4,Malyshev5} some graph parameters (the density and packing number) were considered and it was described how its growth in terms of the number of vertices  affects on the complexity of the independent set problem. A theory, when an NP-complete graph problem becames easier, is investingated appling to the family of hereditary classes in the papers \cite{AlekseevMalyshev1,AlekseevMalyshev2,Malyshev1,Malyshev2,Malyshev3,Malyshev6,Malyshev7,Malyshev8}. Our main interest is to determine a dependence of the integer programming problem complexity on spectrum of sub-determinants of the restriction matrix. 

\section{ The polytopes with bounded determinants }

The main result of this section is 

\begin{theorem}\label{StrictKMod_PolyWidthT}
Let $A \in \mathbb{Z}^{m \times n}$, $b \in \mathbb{Z}^m$ and $P = P(A,b)$ be a polytope. If $w(P) > (\Delta_{lcm}(A)-1)\cfrac{\Delta(A)}{\Delta_{gcd}(A)}(n+1)$, then $|P(A,b) \cap \mathbb{Z}^{n}| \geq n + 1$. We can find an integer point in $P(A,b) \cap \mathbb{Z}^{n}$ by a polynomial-time algorithm.
\end{theorem}

This theorem has two trivial corollaries:

\begin{corollary}\label{mainCol1}
Let $A \in \mathbb{Z}^{m \times n}$ be a \emph{$k$-modular} matrix, $b \in \mathbb{Z}^m$  and $P = P(A,b)$ be a polytope.

If $w(P) > (\Delta(A)-1)\cfrac{\Delta(A)}{\delta(A)}(n+1)$, then $|P(A,b) \cap \mathbb{Z}^{n}| \geq n + 1$. We can find an integer point in $P(A,b) \cap \mathbb{Z}^{n}$ by a polynomial-time algorithm.
\end{corollary}

\begin{corollary}\label{mainCol1}
Let $A \in \mathbb{Z}^{m \times n}$, $b \in \mathbb{Z}^m$, $P = P(A,b)$ be a polytope and each $r(A) \times r(A)$ sub-determinant of $A$ is equal to $\pm \Delta(A)$ or $0$. 

If $w(P) > (\Delta(A)-1)(n+1)$, then $|P(A,b) \cap \mathbb{Z}^{n}| \geq n + 1$. We can find an integer point in $P(A,b) \cap \mathbb{Z}^{n}$ by a polynomial-time algorithm.
\end{corollary}

We need the following  lemmas to prove the theorem \ref{StrictKMod_PolyWidthT}.

\begin{lemma}\label{Lm1} Let $B \in \mathbb{Z}^{n \times s}, \ M=cone(B)\cap \mathbb{Z}^n \setminus \{0\}$,
and $c^{\top} x^*=min_{x\in M}{c^{\top} x}$, then $x^* \in conv(0\,B)$.
\end{lemma}
\begin{proof}
  Suppose, that $c^{\top} B_{* \, l} = min _i \, c^{\top} B_{* \, i}$ and $x^*= B t \notin conv(0,b_1,...,b_s)$, then
$t_1+\cdots+ t_s>1$.  Therefore $c^{\top} x^*=c^{\top} B t \ge (t_1+\cdots +t_s)c^{\top} B_{* \, l} > c^{\top} B_{* \, l}$, this is a contradiction.
\end{proof}

\begin{lemma}\label{Lm2}
Let $z'=b'^{\top} x' =min\{b'^{\top} x:\: x \in M\},\ z''=b''^{\top} x'' =min\{b''^{\top} x:\: x \in M\}$, then $|z'-z''|\le
||b'-b''||_\infty||x''||_1$.
\end{lemma}
\begin{proof} 
Let $z' \ge z''$, then $z'-z''\le b'^{\top} x''-b''^{\top} x''\le ||b'-b''||_\infty||x''||_1$.
\end{proof}

\begin{lemma}\label{Lm3}
Let $P'=P(A,b')$ and $P''=P(A,b'') $ be nonempty polytopes, then $|w(P')-w(P'')|\le \cfrac{\Delta(A)}{\Delta_{gcd}(A)}(n+1)||b'-b''||_\infty$ .
\end{lemma}
\begin{proof}
$ w(P')=min\{max_{P'} c^\top x - min_{P'} c^\top x:\ c\in \mathbb{Z}^n\setminus\{0\} \}$. According to the duality theorem of the linear programming 
$$max_{P'} c^\top x=min\{b'^\top y:\ A^\top y=c,\ y\ge 0\}$$ $$ min_{P'} c^\top x=-min\{b'^\top z:\ -A^\top z=c,\ z\ge 0\}.$$ Therefore, 
$$w(P')=min\{b'^\top y+b'^\top z:\ A^\top y+A^\top z=0,\ y \ge 0,\ z\ge 0,\ A^\top y \in \mathbb{Z}^n\setminus \{0\}\}. $$
Let us consider the cone $C=\{\binom{y}{z} \in R^{2m} :\: A^\top y+A^\top z=0,\ \binom{y}{z} \ge 0\}$. From the equation $ A^\top y+A^\top z=0$ and the Kramer's rule, it follows that $C$ is generated by the vectors with components in the set of all $r(A) \times r(A)$ sub-determinants of $A$. Therefore, the maximal absolute value of coordinates of these vectors is at most $\Delta(A)$. This value can be decreased by a dividing each of components by $\Delta_{gcd}(A)$.
According to Carateodory's theorem, every vector has at most $n+1$ non-zero coordinates. So, by the lemma \ref{Lm1}, there are $y',z'$ such that $w(P')= b'^\top y'+b'^\top z'$ and $\sum_{i = 1}^m (y'_i+z'_i)\le \cfrac{\Delta(A)}{\Delta_{gcd}(A)} (n+1)$. Similarly, we can show that $w(P'')= b''^\top y''+b''^\top z''$ and $\sum_{i = 1}^m (y''_i+z''_i)\le \cfrac{\Delta(A)}{\Delta_{gcd}(A)} (n+1)$ for  some $y'',z''.$ Suppose, for clearness, that 
$w(P')\ge w(P'')$. We have $|w(P')-w(P'')| = (b'^\top y'+b'^\top z')- (b''^\top y''+b''^\top z'') \le (b'^\top y''+b'^\top z'')- (b''^\top y''+b''^\top z'') = (b'-b'')^\top y''+(b'-b'')^\top z''\le ||b'-b''||_\infty||(y'',z'')||_1\le \cfrac{\Delta(A)}{\Delta_{gcd}(A)} (n+1)||b'-b''||_\infty$.
\end{proof}

Now, we are ready to prove the main result (theorem
\ref{StrictKMod_PolyWidthT}) of this section.\\

\begin{proof}
Let $b^\prime \in
\mathbb{Z}^{m}$, such that $b_i^\prime = b_i - (b_i \mod
\Delta_{lcm}(A))$ and $P^\prime = P(A,b^\prime)$. If $w(P) >
(\Delta_{lcm}(A) - 1)\cfrac{\Delta(A)}{\Delta_{gcd}(A)}(n+1)$, then by the lemma \ref{Lm3}, $w(P^\prime) >0$. Thus, $P^\prime$ is full-dimensional and each component
of $b^\prime$ is divided by $\Delta_{lcm}(A)$. So, it is easy to see,
that $P^\prime$ is a full dimensional polytope and all components of any vertex of $P^\prime$ are integer. Now, the theorem follows from the fact that $P^\prime \subseteq P$.

We can use any polynomial algorithm of linear programming (Khachiyan's algorithm \cite{Panos,Khachiyan}) to find some vertex of $P^\prime$ as an integer point of $P$.
\end{proof}

\section{ The simplices with bounded determinants }

A part of this section describes results of the R.~E.~Gomory \cite{Gomory_Relation,Gomory_IntFaces,Gomory_Comb,Hu}. We will repeat some of the Gomory's arguments, slightly modifying them for our purposes.

Let $A \in \mathbb{Z}^{n \times n}$, $B \in \mathbb{Z}^{n \times s}$, $b \in \mathbb{Z}^s$, and $|det(A)| = \Delta > 0$. Consider the system

\begin{align}\label{GomorySystem}
\begin{cases}
A x + B y = b \\
x \in \mathbb{Z}^n,\, y \in \mathbb{Z}^s_+
\end{cases}.
\end{align}

Let $D$ be the Smith normal form \cite{Schrijver1998} of the matrix $A$, then $A = P^{-1} D Q^{-1}$, where $P^{-1},\, Q^{-1}$ are integer unimodular matrices. So system \eqref{GomorySystem} becomes

\begin{align*}
\begin{cases}
D Q^{-1} x + P B y = P b \\
x \in \mathbb{Z}^n,\, y \in \mathbb{Z}^s_+
\end{cases}.
\end{align*}

After the unimodular map $Q^{-1} x \to x$ and removing $x$'s variables the system becomes

\begin{align}\label{GomoryGroupSystem}
\begin{cases}
P B y \equiv P b\, (mod\, D) \\
y \in \mathbb{Z}^s_+
\end{cases}.
\end{align}

There is a bijection between variables $x$ and $y$ giving by the formula $x = A^{-1}( b - B y)$.

Let $M(A,\,B,\,b)$ be a polyhedron induced by the system \eqref{GomoryGroupSystem}.

\begin{definition} [Gomory]
Let $y \in M(A,\,B,\,b)$. Hence $P B y \equiv P b \, (mod\, D)$.
We say that $y$ is an \emph{irreducible} point of $M(A,\,B,\,b)$ if for any $u \not= v$, such that $u \leq y$ and $v \leq y$, we have $P B u \not\equiv P B v \, (mod\, D)$.
\end{definition}

\begin{theorem} [Gomory]
Columns of the matrix $P B$ induce an additive group, the group operation is an addition by the modulo $D$. The power of this group is at most $\Delta$. If $M(A,\,B,\,b) \not= \emptyset$ and $y$ is an \emph{irreducible} point of $M(A,\,B,\,b)$, then $\prod_{k=1}^{s}(1 + y_k) \leq \Delta$.
\end{theorem}

\begin{proof}
It is easy to see that columns of $P B$ induce an additive group. Let $g$ be element of this group, so $g = P B t \mod D$ for some $t \in \mathbb{Z}^s_+$. Hence, $0 \leq g_k \leq D_{k\,k}-1$, where $1 \leq k \leq n$. So, a total number of group elements is at most $\prod_{k=1}^n D_{k\,k} = \Delta$.

Let $y$ be an \emph{irreducible} point of $M(A,\,B,\,b)$, and $t \leq y$. From the definition of the point $y$ it follows that all group elements $g = P B t  \mod  \Delta$ are different for different $t$, so different combinations of $t$ induce different elements of a group. Since the total number of $t$ combinations is $\prod_{k=1}^s (1 + y_k)$ and the number of distinct group elements is at most $\Delta$, then we have $\prod_{k=1}^s (1+y_k) \leq \Delta$.
\end{proof}

\begin{theorem}[Gomory] \label{GomoryVertTh}
Let $y \in vert(M(A,\,B,\,b))$, then $y$ is an \emph{irreducible} point of $M(A,\,B,\,b)$.
\end{theorem}

\begin{corollary}[Gomory]
Let $y \in vert(M(A,\,B,\,b))$, then $\prod_{k=1}^s (1+y_k) \leq \Delta$.
\end{corollary}

\begin{lemma}\label{GomoryExistenceLm}
Let $I$ be the identity matrix. There is a polynomial time algorithm to find the point $y \in M(A,\, I,\, b)$ with property $\sum_{k=1}^s y_k \leq \Delta-1$.
\end{lemma}

\begin{proof}

The system for $M(A,\, I,\, b)$ has very simple structure:
\begin{align*}
\begin{cases}
P y \equiv P b\, (mod\, D) \\
y \in \mathbb{Z}^s_+
\end{cases}.
\end{align*}

So $y = b + P^{-1} D t$, for any $t \in \mathbb{Z}^s$. Let $P^{-1}D = H Q$, where $H$ be the Hermite normal form and $Q$ be an unimodular matrix. Then $y = b + H t$, for any $t \in \mathbb{Z}^s$, and it is easy to see that we can choose $y_k \leq H_{k\,k}$. Hence, we have the vector $y \in M(A,\, I,\, b)$ with the property $\prod_{k=1}^s y_k \leq \prod_{k=1}^s (H_{k \, k} - 1)$. We know that $\prod_{k=1}^{s} H_{k \, k} = \Delta$ and a maximum of the sum $\sum_{k=1}^s H_{k \, k}$, that equals to $\Delta + s -1$ is reached when $H_{k \, k} = 1$ for $k \in [1,\,s-1]$ and $H_{s \, s} = \Delta$. Hence, the maximum of the sum $\sum_{k=1}^s (H_{k \, k} - 1)$ is
$\Delta - 1$. 

Moreover, finding the Smith and Hermite normal forms are polynomial solvable problems
\cite{Schrijver1998}.

\end{proof}

Now, we can prove the main result of this section.

\begin{theorem}\label{SimplexKWidthTh}
Let $A \in \mathbb{Z}^{(n+1) \times n}$, $b \in \mathbb{Z}^{n+1}$, $P=P(A,b)$ be a simplex. If $w(P(A,b)) \geq \delta(A)-1$, then $P(A,b) \cap \mathbb{Z}^{n} \not= \emptyset$. There is polynomial-time algorithm to find some integer point in $P$. 
\end{theorem}

\begin{proof} Suppose that $\delta(A) > 1$ (the case of $\delta(A) = 1$ is trivial).
Let $\hat A x \leq \hat b$ be the subsystem of the system $A x \leq b$, where $\hat A \in \mathbb{Z}^{n \times n}$ and $\hat b \in \mathbb{Z}^n$, such that $|det(\hat A)| = \delta(A)$. Let $C = \{x \in \mathbb{R}^n :\: \hat A x \leq \hat b\}$. Then $P \subset C$.
Let $v \in \mathbb{Q}^n$ be the vertex of $P$, such that $\hat A v = \hat b$, so $v = {\hat A}^{-1} {\hat b}$.

Let $B = -(\delta(A)-1) {\hat A}^{-1}$, then from elementary theory of polyhedrons \cite{Schrijver1998,Ziegler} it follows that $C = v + cone(B)$. So, the columns of $B$ define edges of the cone $C$. Let $S = v + conv(0\: B)$. So $S \subset C$.

Additionally, from elementary theory of polyhedrons it follows that $n+1$ lines of the system $A x \leq b$ correspond to $n+1$ facets of $P$ by the following way: if $F$ is a facet of $P$, then $F = P \cap \{x \in \mathbb{R}^n :\: A_{k\,*} x = b_k\}$, for some line $(A_{k\,*}\: b_k)$ of the system $A x \leq b$.

Let $F_1,\, F_2,\, \dots,\, F_n$ be the faces of $P$ corresponding to the system $\hat A x \leq \hat b$, and let $F_{n+1}$ be the last facet of $P$ corresponding to the line of the system $A x \leq b$ that is not included to the system $\hat A x \leq \hat b$.

First, we need to prove that $S \cap \mathbb{Z}^n \not= \emptyset$. The set $C \cap \mathbb{Z}^n$ is induced by the following equivalent systems:

\begin{align}\label{CornerSystem}
\begin{cases}
\hat A x \leq \hat b \\
x \in \mathbb{Z}^n \\
\end{cases}
\sim&
\begin{cases}
\hat A x + y = \hat b \\
x \in \mathbb{Z}^n \\
y \in \mathbb{Z}_+^n \\
\end{cases}
\end{align}

By the lemma \ref{GomoryExistenceLm}, there is a polynomial-time algorithm to find a solution $y^* \in \mathbb{Z}_+^n$ of this system with the property that $\sum_{k=1}^n y^*_k \leq \delta(A) - 1$. Hence, $\exists x^* \in \mathbb{Z}^n$ that $\hat A x^* + y^* = \hat b$ and $x^* = {\hat A}^{-1} \hat b - {\hat A}^{-1} y^*$. Finally, $x^* = v + B \cfrac{1}{\delta(A)-1} y^*$, it is equivalent to the statement that $x^* \in S$.

To finish, we need to prove that $S \subseteq P$. Suppose that $S \not\subseteq P$. Let $g^{(1)},\, g^{(2)},\, \dots,\, g^{(n)}$ be the vertices of $S$ that are adjacent to the vertex $v$. If $S \not\subseteq P$, then the facet $F_{n+1}$ intersects some edge $[v,\,g^{(k)}] \in S$, such that $[v,\,g^{(k)}] \cap F_{n+1} = \{u\}$, $u \in vert(P)$ and $u \not= g^{(k)}$. Let $F_j$, where $1 \leq j \leq n$, be the opposite facet to the vertex $u$ and $(A_{j\,*} \: b_j)$ be the line of the initial system correspondent to $F_j$. Then, $|max\{ A_{j\,*}x :\: x \in P \} - min\{ A_{j\,*}x :\: x \in P \}| = |b_j - A_{j\,*} u | < |b_j - A_{j\,*} g^{(k)} | = \delta(A)-1$. The last statement contradicts to the theorem assumption that $w(P) \geq \delta(A)-1$.
\end{proof}

\begin{corollary}
Let $A \in \mathbb{Z}^{(n+1) \times n}$, $b \in \mathbb{Z}^{n+1}$, $c \in \mathbb{Z}^n$ $P=P(A,b)$ be a simplex. Let the vertex $v$ be an optimal solution of the linear problem $max\{ c^\top x :\: x \in P \}$, $J = \{j :\: A_{j\,*} v = b_j \}$ and $\Delta = |det(A_{J\,*})|$. If $w(P(A,b)) \geq \Delta-1$, then $P(A,b) \cap \mathbb{Z}^{n} \not= \emptyset$, and there is an algorithm with the computational complexity $O(n \Delta)$ that solves the integer problem $max\{ c^\top x :\: x \in P \cap \mathbb{Z}^n \}$.
\end{corollary}

\begin{proof}
We have already proved that $S \subseteq P$, where $S = v + conv((0\:B))$ and $B = - (\Delta-1) {\hat A}^{-1}$ (see the previous Theorem). It is easy to see that $max\{ c^\top x :\: x \in P \cap \mathbb{Z}^n \} = max\{ c^\top x :\: x \in S \cap \mathbb{Z}^n \}$. Additionally, from Theorem \ref{GomoryVertTh}, it follows that all vertices of the system \eqref{CornerSystem} are in the set $S$. But, we already have a remarkable algorithm proposed by Gomory, Hu \cite{Gomory_Relation,Hu,HuGroupMin} for an integer optimization in systems of the type \eqref{GomorySystem} or \eqref{CornerSystem} with the computational complexity $O(n \Delta)$.
\end{proof}

\begin{corollary}
Let $A \in \mathbb{Z}^{(n+1) \times n}$, $b \in \mathbb{Z}^{n+1}$, $c \in \mathbb{Z}^n$ $P=P(A,b)$ be a simplex. If $w(P(A,b)) \geq \Delta(A)-1$, then $P(A,b) \cap \mathbb{Z}^{n} \not= \emptyset$ and there is an algorithm with the computational complexity $O(n \Delta)$ that solves the integer problem $max\{ c^\top x :\: x \in P \cap \mathbb{Z}^n \}$.
\end{corollary}

\section{Integer Programs with Almost Unimodular Matrices}

We will use Theorem \ref{MainTh1} to prove the following result.
\begin{theorem}\label{AlUnTh}
Let $A \in \mathbb{Z}^{m \times n}$  be \emph{almost unimodular}, $b \in
\mathbb{Z}^m$, $c \in \mathbb{Z}^n$. Then, there is a
polynomial-time algorithm to solve $max\{c^\top x :\:  x
\in P(A,b) \cap \mathbb{Z}^n \}$.
\end{theorem}
\begin{proof}
Let $P = P(A,b)$. We can check non-emptiness of $P \cap \mathbb{Z}^n$ by a polynomial time algorithm using Theorem \ref{MainTh1}. Then, if $P \cap \mathbb{Z}^n \not= \emptyset$, we need to find a vertex $v$ of $P$, that is an optimal solution of the relaxed problem $max\{c^\top x : A x \leq b,\  x \in \mathbb{Q}^n \}$. To this end, one can use the  Khachiyan's  polynomial algorithm for linear programming \cite{Khachiyan}. If $v$ is integral, then it is an optimal solution of the integer problem. If not, we consider the shifted cone $N(v) = \{ x \in \mathbb{R}^n :\: A_{J\,*} x \leq b \}$, where $J = \{k :\: A_{k\,*} v = b \}$.

By Theorem \ref{MainTh1}, $max\{c^\top x : A x \leq b,\  x \in \mathbb{Z}^n \} = max\{ c^\top x :  x \in N(v),\, x \in \mathbb{Z}^n \}$ and an optimal integral point $v^*$ lies on some edge of $N(v)$.

Let $\overline{B}_k$ be the set $N(v) \cap \{ x \in \mathbb{R}^n :\: x_k = \lceil v_k \rceil \}$ and $\underline{B_k}$ be the set $N(v) \cap \{ x \in \mathbb{R}^n :\: x_k = \lfloor v_k \rfloor \}$.

Now, we are going to prove that there exists $k$, $1 \leq k \leq n$ such that $v^*$ is the optimal solution of the linear problem $max\{c^\top x :\: x \in \overline{B}_k \}$ or problem $max\{c^\top x :\: x \in \underline{B_k} \}$.

Let $L$ be an edge of $N(v)$. If $|L \cap \overline{B}_k| = 1$ or $|L \cap \underline{B_k}| = 1$ for some $k$ and we know that $L \cap \overline{B}_k = {u}$ or $L \cap \underline{B_k} = {u}$, then $u \in \mathbb{Z}^n$. Indeed, assume that $L$ is induced by an integer subsystem of the system $Ax \leq b$ denoted by $\hat A x \leq \hat b$. Additionally, we can assume that $\hat A$ is $(n-1)\times n$ matrix. Then, the point $u$ is induced by the system $\dbinom{\hat A}{0^{k-1}\: 1\: 0^{n-k}} x = \dbinom{\hat b}{\sigma}$, where $\sigma$ is $\lceil v_k \rceil$ or $\lfloor v_k \rfloor$. By assumptions of the theorem, the matrix $\dbinom{\hat A}{0^{k-1}\: 1\: 0^{n-k}}$ is unimodular, so $u$ is integer.

Hence, we have only two possibilities for $vert(\overline{B}_k)$ (the same is true for $vert(\underline{B_k})$ ):

1.  $vert(\overline{B}_k) = \{v\}$. This is the case, when $v \in \{ x \in \mathbb{R}^n :\: x_k = \lceil v_k \rceil \}$, so no new vertices were created.

2. $vert(\overline{B}_k) \subset \mathbb{Z}^n$. Indeed, if $v \notin \{ x \in \mathbb{R}^n :\: x_k = \lceil v_k \rceil \}$, then new vertices must be generated. Each vertex of $\overline{B}_k$ is the result of an intersection of some edge of $N(v)$ with the hyperplane $\{ x \in \mathbb{R}^n :\: x_k = \lceil v_k \rceil \}$. By the previous, all these intersections are integer.

Now, let $L^*$ be an edge of $N(v)$ that contains an optimal point $v^*$.

There are few possible cases for some fixed $k$:

1. $L^* \subseteq \overline{B}_k$. Then, it is easy to see that an optimal solution of the linear problem $max\{c^\top x :\: x \in \overline{B}_k \}$ is a point $v$.

2. $L^* \cap \overline{B}_k = \emptyset$. In this case, the linear problem $max\{c^\top x :\: x \in \overline{B}_k \}$ can be inconsistent or it can have an integer or a rational solution. By the previous, if a solution is rational, then it is the vertex $v$.

3. $|L^* \cap \overline{B}_k| = 1$. This is the case of an intersection of a ray and a hyperplane. More precisely, we can prove that $L^* \cap \overline{B}_k = \{v^*\}$. Let $L^* \cap \overline{B}_k = \{u\}$ and $u \not= v^*$. By the previous results, the point $u$ must be integral. Hence $v^* \in [v,\,u]$, because $L^*$ is an edge of $N(v)$, and the objective function only can be increased by the edge $L^*$. But $[v,\,u]$ can't have any other integer points except $u$ because $u_k = \lceil v_k \rceil$ or $u_k = \lfloor v_k \rfloor$. So $u = v^*$.

4. The same items are true for $L^* \cap \underline{B_k}$.

\medskip

Since $L^*$ contains $v^* \in \mathbb{Z}^n$, then $\exists k$, $1 \leq k \leq n$, such that $|L^* \cap \overline{B}_k| = 1$ or $|L^* \cap \underline{B_k}| = 1$. Hence, by the previous, we have $L^* \cap \overline{B}_k = \{v^*\}$ or $L^* \cap \underline{B_k} = \{v^*\}$. If this is true, then $v \notin \{ x \in \mathbb{R}^n :\: x_k = \lceil v_k \rceil \} \cup \{ x \in \mathbb{R}^n :\: x_k = \lfloor v_k \rfloor \}$ and all vertices of $\overline{B}_k$ or $\underline{B_k}$ resp. are integer.

In conclusion, we note that our final algorithm consists one running of a linear programming algorithm to find the vertex $v$ as an optimal solution of the linear problem $max\{ c^\top x :\: A x \leq b\}$, and $2 n$ times running of a linear programming algorithm to solve the problems $max\{ c^\top x :\: x \in \overline{B}_k \}$ and $max\{ c^\top x :\: x \in\underline{ B_k} \}$, for each $1 \leq k \leq n$. One of this solutions must be $v^*$ that can be recognized by integrality and maximality of the objective function.

Since any linear programming problem can be solved by a polynomial-time algorithm due to Khachiyan \cite{Khachiyan}, our algorithm is polynomial too.
\end{proof}

Let us give some generalization of this proof. The matrix $A$ is called \emph{$k$-almost unimodular}, for $1 \leq k \leq n-1$, if $\Delta(A) \leq 2$ and any $(r(A)-k)\times(r(A)-k)$ sub-matrix has a determinant from the set $\{0,\pm 1\}$.

\begin{theorem}\label{AlUnTh}
Let $A \in \mathbb{Z}^{m \times n}$  be \emph{$k$-almost unimodular}, $b \in
\mathbb{Z}^m$, $c \in \mathbb{Z}^n$. Then, there is a
polynomial-time algorithm for fixed $k$ to solve $max\{c^\top x :\: x \in P(A,b) \cap \mathbb{Z}^n \}$.
\end{theorem}

\begin{proof}
To proof this theorem, we need to consider sets $\overline{B}_J = N(v) \cap \{x \in \mathbb{R}^n :\: x_J = \lceil v_J \rceil \}$ and $\underline{B}_J = N(v) \cap \{x \in \mathbb{R}^n :\: x_J = \lfloor v_J \rfloor \}$, for all subsets $J \subset \overline{1,n}$, $|J| = k$. There are $\binom{n}{k}$ of these subsets.

Again, we have following possibilities for $\overline{B}_J$:

1. $vert(\overline{B}_J) = \{v\}$.

2. $v^* \notin vert(\overline{B}_J) \subset \mathbb{Z}^n$.

3. $v^* \in vert(\overline{B}_J) \subset \mathbb{Z}^n$.

4. The same items are true for $\underline{B}_J$.

Since $v^* \in \mathbb{Z}^n$, then $\exists J$, such that $v^* \in vert(\overline{B}_J)$ or $v^* \in vert(\underline{B}_J)$. If this is true, then $v \notin \{ x \in \mathbb{R}^n :\: x_J = \lceil v_J \rceil \} \cup \{ x \in \mathbb{R}^n :\: x_J = \lfloor v_J \rfloor \}$, so the linear programs $max\{ c^\top x :\: x \in \overline{B}_J\}$ or $max\{ c^\top x :\: x \in \underline{B}_J\}$ have $v^*$ as an optimal solution.

Our final algorithm consists one running of a linear programming algorithm to find the vertex $v$, and $2 \binom{n}{k}$ times running of a linear programming algorithm to solve the problems $max\{ c^\top x :\: x \in \overline{B}_J \}$ and $max\{ c^\top x :\: x \in\underline{ B_J} \}$, for each $J \subset \overline{1,n}$, $|J| = k$. One of this solutions must be $v^*$ that can be recognized by integrality and maximality of the objective function.
\end{proof}

\section{Examples of the Corner Polyhedrons with an Exponential Number of Edges}

Theorem \ref{MainTh1} implies that there is a polynomial-time algorithm for solving the integer program $max\{c^\top x :\: x \in P(A,b) \cap \mathbb{Z}^n \}$, when $\Delta(A) \leq 2$ and the number of edges in the shifted cone $N(v)$ is bounded by a polynomial, where $v$ is an optimal solution of the relaxed linear problem. But, there is an example of an exponential number of edges in the shifted cone $P(A,b)$ with $\Delta(A) \leq 2$.

Consider the set $D = conv\{ {\{0,\,1\}}^{n} \} = \{x \in \mathbb{R}^n : 0 \leq x_i \leq 1,\ \ i \in \overline{1,n} \}$.

We use the homogenization and find the matrix $A_D$, such that: $\{ \binom{x}{x_0} : A_D \binom{x}{x_0} \leq 0 \} = \{ \binom{x}{x_0} : x \in D,\ x_0 \geq 0 \}$.
Finally, we create the matrix $A_D^{\prime}$ by multiplying any column of $A_D$ by $2$. $A_D$ is unimodular so $A_D^{\prime}$ is bimodular and we have the following theorem:

\begin{theorem}\label{ExpEdTh1}
The cone $\{ x \in \mathbb{R}^{n+1} : A_D^{\prime} x \leq 0 \}$ has an exponential number of edges. Any of its edges can be represented by the bimodular submatrix of $A_D^\prime$.
\end{theorem}

\section{Acknowledgments}
The authors wish to express special thanks for the invaluable assistance to P.M.~Pardalos, A.J.~Chirkov, D.S.~Malyshev, N.U.~Zolotykh and V.N.~Shevchenko.

The work is supported by LATNA Laboratory NRU HSE RF government grant ag. 11.G34.31.0057 and Russian Foundation for Basic Research grant ag. 15-01-06249 A.

\end{document}